\documentclass[11pt]{amsart}
\usepackage{amsmath}
\usepackage[dvips]{epsfig}
\usepackage{slashed}
\usepackage{accents}
\usepackage{color}
\usepackage{leftidx}
\usepackage{verbatim}

\theoremstyle{plain}
\newtheorem{theorem}{Theorem}

\newtheorem{lemma}[theorem]{Lemma}
\newtheorem{corollary}[theorem]{Corollary}

\newtheorem{conjecture}[theorem]{Conjecture}
\newtheorem{remark}[theorem]{Remark}

\DeclareMathOperator{\ric}{Ric}

\DeclareMathOperator{\tr}{tr}
\DeclareMathOperator{\hess}{Hess}
\DeclareMathOperator{\dist}{dist}

\begin{document}
\title[New Restrictions on the Topology of Extreme Black Holes]{New Restrictions on the Topology of Extreme Black Holes}

\author[Khuri]{Marcus Khuri}
\address{Department of Mathematics\\
Stony Brook University\\
Stony Brook, NY 11794, USA}
\email{khuri@math.sunysb.edu}

\author[Woolgar]{Eric Woolgar}
\address{Department of Mathematical and Statistical Sciences\\
University of Alberta\\
Edmonton, AB, Canada T6G 2G1}
\email{ewoolgar@ualberta.ca}

\author[Wylie]{William Wylie}
\address{215 Carnegie Building\\
Department of Mathematics\\
Syracuse University\\
Syracuse, NY 13244, USA}
\email{wwylie@syr.edu}

\thanks{M. Khuri acknowledges the support of NSF Grant DMS-1708798. E. Woolgar was supported by a Discovery Grant RGPIN 203614 from the Natural Sciences and Engineering Research Council. W. Wylie acknowledges the support of Simons Foundation Grant \#355608 and NSF Grant DMS-1654034.}

\begin{abstract}
We provide bounds on the first Betti number and structure results for the fundamental group of horizon cross-sections for extreme stationary vacuum black holes in arbitrary dimension, without additional symmetry hypotheses. This is achieved by exploiting a correspondence between the associated near-horizon geometries and the mathematical notion of $m$-quasi Einstein metrics, in addition to generalizations of the classical splitting theorem from Riemannian geometry. Consequences are analyzed and refined classifications are given for the possible topologies of these black holes.
\end{abstract}
\maketitle

\section{Introduction}
\label{sec1}
\setcounter{equation}{0}
\setcounter{section}{1}

\noindent Stephen Hawking revolutionized the theory of black holes. One part of that legacy is his horizon topology theorem \cite{Hawking,HawkingEllis} which states that, in spacetime dimension $D=4$, cross-sections of the event horizon of asymptotically flat stationary vacuum black holes that obey the dominant energy condition must have spherical topology $S^2$. In higher dimensions such a simple characterization does not hold, as is illustrated by the $D=5$ ring solutions of Emparan-Reall \cite{EmparanReall} and Pomeransky-Sen'kov \cite{PomeranskySenkov} which have cross-sectional horizon topology $S^1\times S^2$. A natural question then arises: what are the possible horizon topologies for asymptotically flat stationary vacuum black holes in dimensions $D>4$?  A primary observation that underpins much of the study of this problem is the fact that such horizons, or more generally stable marginally outer trapped surfaces, are of positive Yamabe type \cite{Galloway,Galloway2,GallowaySchoen,Lucietti}. This means that these manifolds admit Riemannian metrics of positive scalar curvature. In $D=5$ this provides strong restrictions on the possible topologies of the 3-dimensional horizons. More precisely, by the prime decomposition theorem \cite{Hempel} along with a result of Gromov-Lawson \cite{GromovLawson} and resolution of the Poincar\'{e} conjecture, a compact orientable 3-manifold having positive Yamabe invariant is diffeomorphic to a spherical space, $S^1\times S^2$, or a finite connected sum thereof. Here a spherical space refers to a quotient of the 3-sphere $S^3/\Gamma$ where $\Gamma\subset O(4)$ is a discrete subgroup, an example is the lens space $L(p,q)$ in which $\Gamma=\mathbb{Z}_p$.

Another tool used to study the topology of black holes is the topological censorship theorem 
\cite{FriedmanSchleichWitt} which states that any curve beginning and ending in the asymptotically flat end can be deformed continuously to lie entirely within the asymptotic region. An example of a result relying on topological censorship is a refinement of the classification given above when $D=5$. In \cite{HollandsHollandIshibashi} it is assumed that in addition to stationarity a $U(1)$ symmetry is present, which is guaranteed for analytic solutions \cite{HollandsIshibashi0,HollandsIshibashiWald,MoncriefIsenberg}, and it is proven that the horizon is one of several possible quotients of $S^3$ by isometries, or a connected sum between copies $S^1\times S^2$ and lens spaces. If multiple axial symmetries are present, namely the isometry group contains $U(1)^{D-3}$, then Hollands-Yazadjiev \cite{HollandsYazadjiev1} have shown that the only possible horizon topologies are $S^3\times T^{D-5}$, $S^2\times T^{D-4}$, or $L(p,q)\times T^{D-5}$, where $T^{D-5}$ denotes the $(D-5)$-dimensional torus. Note that this amount of axisymmetry is only compatible with asymptotic flatness in dimensions $D=4,5$. Furthermore, it should be pointed out that it is not known whether all of these possible topologies are realized by stationary vacuum solutions, even in dimension five. More precisely, while $S^3$ and $S^1\times S^2$ have been realized by the Myers-Perry \cite{MyersPerry} and ring solutions listed above respectively, the only explicit examples of vacuum lenses \cite{ChenTeo} are known to have conical singularities or possess naked singularities. On the other hand,
vacuum lenses and other configurations have been produced in multitude through abstract existence results for singular harmonic maps \cite{KhuriWeinsteinYamada}, although it is not known at this time whether any of these solutions are void of conical singularities; in addition, there has been some progress in the asymptotically Kaluza-Klein and locally Euclidean cases \cite{KhuriWeinsteinYamada1} as well. In other theories, such as $D=5$ minimal supergravity, geometrically regular asymptotically flat lens black holes have been produced \cite{BreunholderLucietti,KunduriLucietti1,TomizawaNozawa}. The latter examples are supersymmetric and hence extremal.

Cobordism theory has also aided in classifying the topologies of black holes. Recall that two compact manifolds of the same dimension are called cobordant if their disjoint union is the boundary of a compact manifold of one higher dimension. The idea is that, in the asymptotically flat case, horizon cross-sections will be cobordant through a simply connected spacelike hypersurface to a sphere $S^{D-2}$ sitting in the asymptotic end, and this should yield conditions on the possible topologies of the horizon. The lowest dimension for which this line of investigation can yield topological restrictions is $D=6$. In \cite{HelfgottOzYanay} Helfgott-Oz-Yanay combined the methods of cobordism theory with Freedman's classification of 4-manifolds \cite{FreedmanQuinn} and results of Donaldson \cite{DonaldsonKronheimer} to show that if the horizon cross-section is simply connected
then it must be homeomorphic to $S^4$, to a finite connected sum of $S^2\times S^2$'s or $\mathbb{CP}^2 \# \overline{\mathbb{CP}}^2$'s---the 1-point blow-up of $\mathbb{CP}^2$ where $\overline{\mathbb{CP}}^2$ denotes the complex projective plane with opposite orientation. If in addition the horizon is a spin manifold, then the connected sum of complex projective planes may be removed from this list. In \cite{KunduriLucietti0} explicit examples of vacuum near-horizon geometries having horizon cross-sections realizing the topologies $S^2\times S^2$ and $\mathbb{CP}^2 \# \overline{\mathbb{CP}}^2$
have been constructed. Surveys concerning the various mathematical techniques used to study black hole topology may be found in \cite{Galloway1,HollandsIshibashi}.

Although much progress has been made in understanding black hole topology, the methods utilized so far have substantial limitations. Indeed, while the primary observation that horizon cross-sections are of positive Yamabe type provides strong constraints in dimensions $D=4,5$, this condition is considerably more flexible in higher dimensions.
For instance any manifold of the form $S^n\times M$, where $n\geq 2$ and $M$ is compact, admits a metric of positive scalar curvature by scaling the round metric of $S^n$ properly. The purpose of this article is to introduce a new technique into the study of horizon topology in the context of extreme black holes, and to analyze its consequences.
The approach here is based on a relation between the associated near-horizon geometries of these black holes and the mathematical notion of $m$-quasi Einstein metrics (see, e.g., \cite{Villani}) first exploited in \cite{KhuriWoolgar,KhuriWoolgar1}, as well as results emanating from the classical splitting theorem of Riemannian geometry \cite{Petersen}.

\begin{theorem}\label{thm1}
Let $\mathcal{H}$ be a degenerate horizon (cross-section) component of a stationary vacuum spacetime with nonnegative cosmological constant $\Lambda\geq 0$.

\begin{itemize}
\item [(i)] The fundamental group $\pi_{1}(\mathcal{H})$ contains an Abelian subgroup of finite index which is isomorphic to $\mathbb{Z}^k$ with $k\leq D-4$.\smallskip

\item [(ii)] The first Betti number satisfies $b_{1}(\mathcal{H})\leq D-4$. \smallskip

\item [(iii)] If $\Lambda>0$ then $\pi_{1}(\mathcal{H})$ is finite.
\end{itemize}
\end{theorem}

Item $(\mathrm{iii})$ has previously been established in \cite{KhuriWoolgar} by different methods.  As an illustration of the additional topological restrictions Theorem \ref{thm1} places on $\mathcal{H}$, in Section \ref{sec3} we use topological arguments stemming from $(\mathrm{i})$ to refine classification results when $D=5$ to show that $\mathcal{H}$ must be diffeomorphic to a spherical space, $S^1 \times S^2$, or $\mathbb{RP}^3 \# \mathbb{RP}^3$, see Corollary \ref{cor7}. The connected sum of projective spaces can be removed from this list if an additional $U(1)$ symmetry is present.
Furthermore it is also shown that in all dimensions, horizon topologies arising from a `nontrivial' connected sum can be ruled out.

It should be pointed out that the topological restrictions of this theorem also hold for near-horizon geometries which may be studied separately from extreme black holes, and are thus of independent interest. In addition, Theorem \ref{thm1} is valid when matter is present as long as an appropriate energy condition is satisfied. The appropriate energy condition is described in \cite[Inequality (8)]{KhuriWoolgar}, and in the case of perfect fluids it reduces to the dominant energy condition.
Also, the inequality holds for pure electric fields normal to a static degenerate horizon.

\section{Main results}
\label{sec2}
\setcounter{equation}{0}
\setcounter{section}{2}

\subsection{Preliminaries} Consider a stationary black hole spacetime of dimension $D$ satisfying the vacuum Einstein equations
\begin{equation}
R_{\mu\nu}=\Lambda g_{\mu\nu}.
\end{equation}
Our normalization is such that $\Lambda$ is $\frac{2}{(n-2)}$ times the usual cosmological constant.
Stationarity gives an asymptotically timelike Killing field, and generically the rigidity theorem \cite{HollandsIshibashi0,HollandsIshibashiWald,MoncriefIsenberg} yields one or more additional rotational symmetries which altogether produces a Killing field $V$ that is normal to the event horizon. The event horizon is then a Killing horizon and on this surface
\begin{equation}
\nabla_{V}V=\kappa V,
\end{equation}
where the constant $\kappa$ denotes the surface gravity. Near each horizon component Gaussian null coordinates $(r,v,x^i)$ may be imposed such that $V=\partial_{v}$, $r=0$ represents the horizon, and $x^i$ are coordinates on the $D-2$-dimensional compact horizon cross-section $\mathcal{H}$. In the degenerate case when $\kappa=0$, the spacetime metric then has the form \cite{KunduriLucietti}
\begin{equation}
g = 2 dv \left(dr +\frac{1}{2}r^2 F(r,x) dv +rh_i(r,x) dx^i\right)
 + \gamma_{ij}(r,x) dx^i dx^j.
\end{equation}
This allows for a near-horizon limit $\phi_{\epsilon}^{*}g\rightarrow g_{NH}$ as $\epsilon\rightarrow 0$, where the diffeomorphisms $\phi_{\epsilon}$ are defined by
$v\rightarrow \frac{v}{\epsilon}$, $r\rightarrow \epsilon r$. The resulting near-horizon geometry may then be expressed as
\begin{equation}
g_{NH} = 2 dv \left(dr +\frac{1}{2}r^2 F(x) dv +rh_i(x) dx^i\right)
 + \gamma_{ij}(x) dx^i dx^j,
\end{equation}
where $\gamma_{ij}$ is the induced metric on $\mathcal{H}$. As a consequence of the Einstein equations, the near-horizon data $(F,h_i,\gamma_{ij})$ satisfy the near-horizon geometry equations on $\mathcal{H}$
\begin{equation}\label{1}
R_{ij}= \frac{1}{2}h_i h_j -\nabla_{(i}h_{j)}+\Lambda\gamma_{ij},\quad\quad\quad
F= \frac{1}{2}|h|^2-\frac{1}{2}\nabla_{i}h^i+\Lambda.
\end{equation}

Near-horizon geometries are closely related to the notion of $m$-quasi-Einstein metrics studied in the mathematical literature. These are solutions to the equation
\begin{equation}
\label{eq6}
\mathrm{Ric}_{X}^m=\lambda \gamma,
\end{equation}
where the generalized $m$-Bakry-\'Emery-Ricci tensor on $\mathcal{H}$ is given by
\begin{equation}
\label{eq7}
\mathrm{Ric}_{X}^{m}=\ric+\frac{1}{2}\pounds_{X}\gamma-\frac{1}{m}X\otimes X,
\end{equation}
in which $X$ is a 1-form/vector field and $\pounds_X$ is Lie differentiation along $X$. It should be noted that some authors reserve this terminology for the special case when the vector field is a gradient $X=\nabla f$. Clearly a vacuum near-horizon geometry defines an $m$-quasi-Einstein metric for $m=2$, $\lambda=\Lambda$, and $X=h$.

\subsection{Splitting theorem} A classical result of Riemannian geometry known as the splitting theorem \cite{CheegerGromoll} asserts that a complete Riemannian manifold with nonnegative Ricci curvature and containing a \emph{line} (an inextendible curve which minimizes the distance between any two of its points), must isometrically split off a Euclidean factor. From this several topological consequences follow. In order to take advantage of this, however, a version of the splitting theorem  under the
hypothesis of nonnegative $m$-Bakry-\'Emery-Ricci curvature is needed. Indeed, a suitable splitting theorem is known to hold in the case that $X$ is a gradient vector field \cite[Theorem 1.3]{FangLiZhang}. Here we show that the arguments of \cite[Theorem 1.3]{FangLiZhang} extend to the non-gradient case.

\begin{theorem}\label{thm2}
Let $(M,g)$
be a complete connected Riemannian manifold of dimension $n$
admitting a complete $C^1$ vector field $X$.
If $\mathrm{Ric}_{X}^{m}\geq 0$ for some
$m>0$, then $M$ splits isometrically as $\mathbb{R}^k \times N$ where $N$ is a complete Riemannian manifold without a line. Moreover the projection of $X$ onto the $\mathbb{R}^k$-factor vanishes, and $N$ has nonnegative $m$-Bakry-\'Emery-Ricci curvature.
\end{theorem}

One can prove this by repeating the arguments of \cite{FangLiZhang} with $df$ replaced by $X$, once two underlying lemmata are suitably modified. The first lemma is a modification of an inequality on the Laplacian \cite[Lemma 2]{CheegerGromoll} to our setting, and is derived from the second variation formula for arclength of curves. The second lemma is a straightforward identity. Using these results, we follow a standard approach. In the presence of a line we construct Busemann functions and show that they are linear. Level sets of these functions then manifest the desired splitting. To set notation let
\begin{equation}
\label{eq8}
Lu=\Delta u - \nabla_Xu,
\end{equation}
where $\Delta u=\tr \hess u$ is the Laplace-Beltrami operator acting on functions. Moreover let $p\in M$ be fixed and set $\rho(q)=\dist(p,q)$.

\begin{lemma}\label{lemma3}
Let $\ric_X^m\ge \lambda g$ for some $\lambda\ge 0$ and $m> 0$. If $x\in M$ is not in the cut locus of $p$ then
\begin{equation}
\label{eq9}
L \rho(x) \le \frac{n+m-1}{\rho(x)}.
\end{equation}
\end{lemma}

\begin{proof}
Let $\gamma:[0,\rho]\to M$ be a unit speed minimal geodesic connecting $p=\gamma(0)$ to $x=\gamma(\rho)$.
Following \cite[derivation of equation (2.1)]{FangLiZhang}, we recall \cite[Lemma 2]{CheegerGromoll} that the second variation of arclength along a geodesic $\gamma$ away from the cut locus of $p=\gamma(0)$ yields
\begin{equation}
\label{eq10}
\Delta \rho(x)  \le \int_0^{\rho} \left [ \frac{(n-1)}{\rho^2} -\frac{t^2}{\rho^2}\ric({\dot \gamma},{\dot \gamma})\right ] dt\ .
\end{equation}
Using $\ric_X^m\ge \lambda g({\dot \gamma},{\dot \gamma})=\lambda$ and performing some trivial integrals, we obtain
\begin{equation}
\label{eq11}
\begin{split}
\Delta \rho(x) \le &\, \frac{(n-1)}{\rho} -\frac{\lambda\rho}{3} +\int_0^{\rho} \frac{t^2}{\rho^2} \left [ \frac12 \pounds_X g ({\dot \gamma},{\dot \gamma})-\frac{1}{m}  g(X,{\dot \gamma})^2 \right ] dt\\
=&\, \frac{(n-1)}{\rho} -\frac{\lambda\rho}{3} +g(X,{\dot \gamma}) -\frac{1}{\rho^2}\int_0^{\rho} \left [ 2t g(X,{\dot \gamma}) +\frac{t^2}{m} g(X,{\dot \gamma})^2 \right ] dt\\
=&\, \frac{(n-1)}{\rho} -\frac{\lambda\rho}{3} +g(X,{\dot \gamma}) -\frac{1}{\rho^2}\int_0^{\rho} \left [ \sqrt{m} +\frac{ tg(X,{\dot \gamma})}{\sqrt{m}}\right ]^2 dt +\frac{m}{\rho^2} \int_0^{\rho}dt\\
=&\, \frac{(n+m-1)}{\rho} -\frac{\lambda\rho}{3} +g(X,{\dot \gamma}) -\frac{1}{\rho^2}\int_0^{\rho} \left [ \sqrt{m} +\frac{ tg(X,{\dot \gamma})}{\sqrt{m}}\right ]^2 dt\\
\le &\, \frac{(n+m-1)}{\rho} +g(X,{\dot \gamma}) \ ,
\end{split}
\end{equation}
when $\lambda \ge 0$. Then clearly
\begin{equation}
\label{eq12}
\begin{split}
L\rho(x)=&\,\Delta \rho(x) -\nabla_X\rho(x)\le \frac{(n+m-1)}{\rho(x)} +g(X,{\dot \gamma})(x) -\nabla_X\rho(x)\\
\le &\, \frac{n+m-1}{\rho(x)}\ ,
\end{split}
\end{equation}
since along the minimizing curve $\gamma$ we have $g(X,{\dot \gamma})=\nabla_X\rho$.
\end{proof}

When a line $\gamma$ is present, we apply this result to the \emph{Busemann function}
\begin{equation}
\label{eq13}
b^{\gamma}(q) := \lim_{t\to\infty} \left [ t - \dist(q,\gamma(t))\right ].
\end{equation}
Some care needs to be taken, as $b^{\gamma}$ is a priori a non-smooth $1$-Lipschitz function  so $Lb^{\gamma}$ must be interpreted in the appropriate weak  sense. The standard technique in analyzing Busemann functions  is to construct barrier (or support) functions, see \cite[Section 2]{EschenburgHeintze}.   Applying Lemma \ref{lemma3} to the standard Busemann support functions and applying the same reasoning as in  \cite[Lemma 2.1]{FangLiZhang} then shows that $Lb^{\gamma}\ge 0$ in the barrier sense when $\ric_X^m\ge 0$ for some $m>0$.
Furthermore, Lemmata 2.4 and 2.5 of \cite{FangLiZhang} can then be invoked to yield $Lb^{\pm}=0$ for $b^{\pm}(x):=\lim_{t\to\infty} \left [ t-\dist(x,\gamma(\pm t))\right ]$, $t\ge 0$.
The next identity is the second component required to modify the proof of \cite{FangLiZhang}.

\begin{lemma}\label{lemma4}
For any function $u\in C^{3}(M)$ it holds that
\begin{equation}
\label{eq14}
L \left ( |\nabla u|^2 \right ) =2\vert \hess u \vert^2 +2\nabla_{\nabla u} \left ( Lu \right ) +2\ric_X^m (\nabla u, \nabla u) +\frac{2}{m} \left [ X(u)\right ]^2\ .
\end{equation}
\end{lemma}

\begin{proof}
By straightforward manipulations we have
\begin{equation}
\label{eq15}
\begin{split}
L \left ( |\nabla u|^2 \right ) = &\, \Delta \left ( |\nabla u|^2 \right ) -\nabla_X \left ( |\nabla u|^2 \right )\\
=&\, 2\vert \hess u \vert^2 +2g\left(\nabla u ,
\Delta \nabla u -\nabla_X \nabla u \right)\\
=&\, 2\vert \hess u \vert^2 +2g\left(\nabla u , \Delta \nabla u +\ric (\nabla u, \cdot) - \nabla \nabla_X u +\nabla_{\nabla u}X \right ) \\
=&\, 2\vert \hess u \vert^2 +2\nabla_{\nabla u} \left ( Lu \right ) +2\ric(\nabla u, \nabla u) +\pounds_X g (\nabla u, \nabla u)\\
=&\, 2\vert \hess u \vert^2 +2\nabla_{\nabla u} \left ( Lu \right ) +2\ric_X^m (\nabla u, \nabla u) +\frac{2}{m} \left [ X(u)\right ]^2\ .
\end{split}
\end{equation}
\end{proof}

We now have all the tools necessary to establish the splitting theorem.

\begin{proof}[Proof of Theorem \ref{thm2}] If $(M,g)$ contains no line, we are done, so assume otherwise. Then we may use the line to construct the associated Busemann functions as above. We may apply equation \eqref{eq14} with $u=b^{\pm}$, together with the earlier result that $Lb^{\pm}=0$ and the condition $\ric_X^m\ge 0$, to obtain $L\left ( |\nabla b^{\pm}|^2\right ) \ge 2|\hess b^{\pm}|^2\ge 0$. But then the strong maximum principle forces $|\nabla b^{\pm}|^2=const$, so we may normalize $\nabla b^{\pm}$ to have unit length and moreover now we have
\begin{equation}
\label{eq16}
0=L\left ( |\nabla b^{\pm}|^2\right ) \ge 2|\hess b^{\pm}|^2\ge 0.
\end{equation}
Thus, $\nabla b^{\pm}$ are parallel and the functions $b^{\pm}$ are linear, as desired. Level sets of $b^{\pm}$ are totally geodesic and, by the completeness of $(M,g)$ are complete in the induced metric.

Let $N$ be the zero level surface of $b^+(x)$ for some $x$ (we could work equally with $b^-$; in fact the level sets coincide). It now follows that
\begin{equation}
\label{eq17}
F:N\times {\mathbb R}:(p,t)\mapsto e^{t\nabla b^+}(p)=:\phi_t(p)
\end{equation}
is an isometry; see \cite{FangLiZhang} or \cite{EschenburgHeintze}. We may now identify $(N\times {\mathbb R}, g_N\otimes dt^2)$ with $(M,g)$ where $g_N$ is the induced metric, and we have that the Ricci curvature splits as $\ric(g)=0\cdot dt^2\oplus \ric(g_{N})$.

Applying \eqref{eq14} to $u=b^+$ and using that $\nabla b^+$ is parallel as well as $Lb^+=0$ produces
\begin{equation}
\label{eq18}
0=L\left (|\nabla b^+|^2\right )=\ric_X^m(\nabla b^+,\nabla b^+) + \frac{1}{m} \left [ X(b^+)\right ]^2 \ .
\end{equation}
Since $\ric_X^m\ge 0$ and $m>0$, we have $\ric_X^m(\nabla b^+,\nabla b^+)=0=X(b^+)$. Then $X$ is tangent to $N$ and
\begin{equation}
\label{eq19}
\frac12 \pounds_X(\nabla b^+,\nabla b^+)=\nabla_{\nabla b^+} \left ( g (X,\nabla b^+)\right ) =0.
\end{equation}
Furthermore, for any $Y\perp \nabla b^{\pm}$ we have $\pounds_X(Y,\nabla b^{\pm})=g(Y,\nabla_{\nabla b^{\pm}}X)=0$ by an easy calculation using the splitting. So the condition $\ric_X^m(g)\ge 0$ descends to $\ric_X^m(g_N)\ge 0$.

If $(N, g_N)$ does not contain a line, then we have obtained the desired splitting.  If $(N,g_N)$ does contain a line then we can apply the splitting to $N$ and split off a Euclidean factor of $N$.  Applying the splitting iteratively,  after finitely many steps we obtain the isometric splitting of $M$ as $\mathbb{R}^k \times N$ where $N$ does not contain any lines.
\end{proof}

Now that Theorem \ref{thm2} has been proven we make some remarks about  Bakry-\'Emery Ricci curvature and $m$-Quasi Einstein metrics.  First note that without further assumptions Theorem \ref{thm2} is not true, even in the gradient case, when the parameter $m$ is negative or infinite,  where we interpret $\mathrm{Ric}_X^{\infty} = \mathrm{Ric} + \frac 12 \pounds_X g$.  A splitting theorem can be proven in cases when $m=\infty$ or $m<0$ if an additional energy condition is placed on the vector field $X$, see \cite[Theorem 6.3]{Wylie}.

We also point out that while the splitting theorem holds for non-gradient vector fields, there are a number of results for $m$-quasi Einstein metrics that require $X=\nabla f$. One that is relevant for the considerations in this paper is a result of \cite{KimKim} and of \cite{ChruscielReallTod} which implies that if a compact manifold admits $\mathrm{Ric}_X^m=0$, for some $m>0$ and $X = \nabla f$, then $X \equiv 0$ and $(M,g)$ is Ricci flat.   Vacuum near-horizon geometries with zero cosmological constant on spherical spaces of dimensions 2 and 3 (arising from the extreme Kerr and Myers-Perry solutions) show that the assumption that $X$ is gradient is necessary in this result, since $S^2$ and $S^3$ do not admit Ricci flat metrics. This may be interpreted as illustrating how the non-gradient case is more flexible, as expected. The fact that the splitting theorem still holds in this setting is then somewhat surprising. In addition, to our knowledge it is not known if there is a compact manifold which admits $\mathrm{Ric}_X^m \geq 0$ but does not support a metric of nonnegative Ricci curvature.  This indicates that it could be true that vacuum near-horizon geometries always admit metrics of non-negative Ricci curvature.


\subsection{Proofs of main results}

\begin{proof}[Proof of Theorem \ref{thm1} $(\mathrm{i})$]When the cosmological constant $\Lambda\geq 0$, Theorem \ref{thm2} may be applied to the universal cover $\widetilde{\mathcal{H}}$ of the horizon cross-section to show that $\widetilde{\mathcal{H}}=\mathbb{R}^k \times\widetilde{N}$, where $\widetilde{N}$ is compact \cite[Theorem 69]{Petersen}. As a consequence, the subgroup $\mathcal{G}\leq\pi_{1}(\mathcal{H})$ of isometries of $\widetilde{N}$ is finite. Hence the kernel of the homomorphism $\pi_{1}(\mathcal{H})\rightarrow\mathcal{G}$ is a subgroup of finite index, and acts discretely as well as cocompactly on $\mathbb{R}^k$ so that it is a crystallographic group. By a well-known theorem of Bieberbach \cite{Wolf}, any such group must contain an Abelian subgroup isomorphic to $\mathbb{Z}^k$ of finite index.

A priori $k\leq \dim\mathcal{H}=D-2$, however this may be refined further. Suppose that $k=D-2$, then $\widetilde{\mathcal{H}}=\mathbb{R}^{D-2}$ and $\mathcal{H}=\widetilde{\mathcal{H}}/\pi_{1}(\mathcal{H})$ is flat. Standard arguments \cite{Petersen} then show that the inclusion $\mathbb{Z}^{D-2}\hookrightarrow\pi_{1}(\mathcal{H})$ is an isomorphism, and thus $\mathcal{H}$ is a torus. However tori do not admit metrics of positive scalar curvature \cite{GromovLawson}, and therefore cannot be of positive Yamabe type, yielding a contradiction. Now suppose that $k=D-3$. In this case $\widetilde{N}$ is a 1-dimensional compact manifold, and hence must be $S^1$. This, however, contradicts the fact that $\widetilde{\mathcal{H}}$ is simply connected. It follows that $k\leq D-4$.
\end{proof}

\begin{proof}[Proof of Theorem \ref{thm1} $(\mathrm{ii})$]
By part $(\mathrm{i})$ there is a subgroup $\mathbb{Z}^k\leq \pi_{1}(\mathcal{H})$ of finite index, with $k\leq D-4$. Since the first homology group $H_{1}(\mathcal{H},\mathbb{Z})$ is isomorphic to the abelianization of the fundamental group, this subgroup is also of finite index $H_{1}(\mathcal{H},\mathbb{Z})$. Therefore the rank of the torsion free part (first Betti number) satisfies $b_1(\mathcal{H})=k\leq D-4$.
\end{proof}

\begin{proof}[Proof of Theorem \ref{thm1} $(\mathrm{iii})$]
If $\Lambda>0$, then the $2$-Bakry-\'Emery-Ricci curvature is strictly positive $\mathrm{Ric}_{h}^{2}> 0$. It is then impossible for the the universal cover to split
as $\widetilde{\mathcal{H}}=\mathbb{R}^k\times\widetilde{N}$ with $k\geq 1$. Thus the universal cover is compact, which again implies that $\pi_{1}(\mathcal{H})$ is finite.
\end{proof}

\subsection{A remark on curvature-dimension conditions}
It is not central to our applications, but for completeness we now show that our `energy condition' form of the curvature-dimension condition is equivalent to a more usual form of the curvature-dimension condition.

\begin{lemma}\label{lemma6} Let $m>0$ and assume that $\ric_X^m\ge \lambda g$ for some $\lambda\in {\mathbb R}$, then
\begin{equation}
\label{eq20}
L\left (|\nabla u|^2\right )\ge 2\nabla_{\nabla u} \left ( Lu \right )+\frac{2}{(n+m)} \left ( Lu \right )^2 +\lambda |\nabla u|^2\ .
\end{equation}
Conversely, if \eqref{eq20} holds for all $u\in C^3(M)$, then $\ric_X^m\ge \lambda g$.
\end{lemma}

\begin{proof}
A direct computation yields
\begin{equation}
\label{eq21}
\begin{split}
\frac{(Lu)^2}{(n+m)}=&\, \frac{(\Delta u)^2}{(n+m)} +\frac{[X(u)]^2}{(n+m)} -\frac{2}{(n+m)}X(u)\Delta u\\
=&\, \frac{(\Delta u)^2}{n} +\frac{[X(u)]^2}{m} -\frac{1}{(n+m)}\left [ \sqrt{\frac{m}{n}} \Delta u + \sqrt{\frac{n}{m}} X(u) \right ]^2\ ,
\end{split}
\end{equation}
so that
\begin{equation}
\label{eq22}
\frac{1}{m}\left [ X(u)\right ]^2= \frac{(Lu)^2}{(n+m)} -\frac{(\Delta u)^2}{n} + \frac{1}{(n+m)}\left [ \sqrt{\frac{m}{n}} \Delta u + \sqrt{\frac{n}{m}} X(u) \right ]^2\ .
\end{equation}
Substituting this into \eqref{eq14} produces
\begin{equation}
\label{eq23}
\begin{split}
L \left ( |\nabla u|^2 \right ) = &\, 2\left \vert \hess u -\frac{1}{n}\left ( \Delta u\right ) g \right \vert^2 +2\nabla_{\nabla u} \left ( Lu \right ) +2\ric_X^m (\nabla u, \nabla u)\\
&\, +\frac{2}{(n+m)}(Lu)^2 + \frac{2}{(n+m)}\left [ \sqrt{\frac{m}{n}} \Delta u + \sqrt{\frac{n}{m}} X(u) \right ]^2\ .
\end{split}
\end{equation}
Equation \eqref{eq20} now follows from the assumption on $\ric_X^m$.

To prove the converse, simply follow the argument of \cite[pp 387--388]{Villani}, noting that no use whatsoever is made in that argument of the assumption $X=\nabla V$.
\end{proof}

\section{Applications}
\label{sec3}
\setcounter{equation}{0}
\setcounter{section}{3}

In this section we will use Theorem \ref{thm1} to refine the topological classification of horizons described in the Introduction. This will be accomplished with the aid of results from geometric group theory. Recall that in a finitely generated group, the smallest number of generators needed to express an element is referred to as the length of the element. Furthermore, such a group is said to have polynomial growth if the number of elements of length at most $\alpha$ is bounded above by a polynomial function $p(\alpha)$, and the minimum degree of any polynomial having this property is the order of growth. An important characterization for groups of polynomial growth was given by Gromov \cite{Gromov,Kleiner}. It asserts that a finitely generated group has polynomial growth if and only if it has a nilpotent subgroup of finite index. In particular, a finitely generated group of exponential growth cannot have an Abelian (more generally nilpotent) subgroup of finite index. Therefore in light of Theorem \ref{thm1} $(\mathrm{i})$, horizon cross-sectional topology cannot take the form of a manifold whose fundamental group is of exponential growth. We note that in the case of a horizon with nonnegative Ricci curvature, a result of Milnor \cite{Milnor} implies directly that the fundamental group is of polynomial growth of order no larger than $D-2$.

The above arguments state that, heuristically, horizons must have limited topology. More precisely, we can rule out some basic constructions that appear in many classifications. Consider the connected sum $M\# N$ of two manifolds $M$ and $N$. The fundamental group of the connected sum is the free product of the individual fundamental groups $\pi_{1}(M\# N)=\pi_1(M)\ast \pi_1(N)$, when the dimension of $M$, $N$ is greater than 2. In general the free product of two nontrivial groups is quite large, and thus should not be able to serve as the fundamental group of a horizon. Indeed, it can be shown \cite[Theorem 4]{Mann} that the free product of two nontrivial groups is always of exponential growth or better, whenever at least one of the two groups making up the free product has order greater than 2. This may be established by showing that such groups contain a non-Abelian free subgroup on two generators. Altogether this establishes the following result.

\begin{theorem}\label{thm3}
Let $\mathcal{H}$ be a degenerate horizon (cross-section) component of a stationary vacuum spacetime with nonnegative cosmological constant $\Lambda\geq 0$. Then $\mathcal{H}$ cannot be expressed as a connected sum $M\# N$ for any compact manifolds $M$ and $N$ both having nontrivial fundamental group, except possibly in the case that $\pi_1(M)=\pi_1(N)=\mathbb{Z}_2$.
\end{theorem}

In dimension $D=5$ this theorem gives a strong refinement of previous horizon topology classifications. Namely, it rules out all possible connected sums of prime 3-manifolds.

\begin{corollary} \label{cor7}
Let $\mathcal{H}$ be a degenerate horizon (cross-section) component of a 5-dimensional stationary vacuum spacetime with nonnegative cosmological constant $\Lambda\geq 0$. Then $\mathcal{H}$ is diffeomorphic to either a spherical space, $S^1\times S^2$, or $\mathbb{RP}^3\# \mathbb{RP}^3$.
\end{corollary}

\begin{proof}
According to Theorem \ref{thm3} and the previous classification, the only possible topology which is not immediately ruled out is $\mathcal{H}=\mathbb{RP}^3\# \mathbb{RP}^3$. In this case $\pi_1(\mathcal{H})=\mathbb{Z}_2 \ast \mathbb{Z}_2$ is the infinite dihedral group, which has an index 2 infinite cyclic subgroup, and thus the statement of Theorem \ref{thm1} $(\mathrm{i})$ cannot be used to exclude this possibility.
\end{proof}

\begin{remark}
In the asymptotically flat or asymptotically Kaluza-Klein case, if there is a $U(1)$ symmetry then $\mathbb{RP}^3\# \mathbb{RP}^3$ may be removed from the statement of Corollary \ref{cor7}.\footnote{
This observation is due independently to Stefan Hollands and James Lucietti.} This follows from \cite[Result 1]{HollandsHollandIshibashi}, which implies that in such a situation the connected sum of two projective spaces can only appear through a connected sum with at least one $S^1\times S^2$. The presence of a $U(1)$ symmetry is generic in the sense that it is assured by the rigidity theorem when the solution is analytic \cite{HollandsIshibashi0,HollandsIshibashiWald,MoncriefIsenberg}. Moreover, in the spherical space case this symmetry further restricts the possible topologies
\cite[Result 1]{HollandsHollandIshibashi}.
\end{remark}

The results above suggest that if a degenerate horizon cross-section is decomposed as a nontrivial connected sum, then at least one member of the sum should be simply connected. Manifolds which cannot be written as a nontrivial connected sum are referred to as prime. Therefore, to a certain extent perhaps horizons of stationary black holes (of any dimension) may be described as almost prime.

\begin{conjecture}\label{conjecture9}
Each connected component of a degenerate
horizon cross-section in an asymptotically flat stationary vacuum spacetime is
`almost prime',\footnote{The related concept of stably prime manifold appears in \cite{KreckLuckTeichner}.} in the sense that if it is expressed as a connected sum of two manifolds then at least one member of the sum must be simply connected.
\end{conjecture}

This may be delicate. For instance, we claim that there exists a solution of the vacuum near-horizon geometry equations with $\Lambda=0$ on $\mathcal{H}=\mathbb{RP}^3 \# \mathbb{RP}^3$, although it is not clear whether this solution arises from the near-horizon limit of an asymptotically flat stationary vacuum spacetime. To verify this claim, observe that the universal cover of the connected sum of two projective spaces is an infinite connected sum of $S^3$'s, which is diffeomorphic to $\mathbb{R}\times S^2$. The group $\mathbb{Z}_2 \ast\mathbb{Z}_2$ acts naturally on this space as follows. The generator for the first $\mathbb{Z}_2$ consists of a reflection in the $\mathbb{R}$ component across 1 together with an application of the antipodal map in the $S^2$ component. The generator for the second $\mathbb{Z}_2$ in the free product is defined similarly, but with a reflection across $-1$. We then have $\mathcal{H}=\mathbb{R}\times S^2/\mathbb{Z}_2\ast\mathbb{Z}_2$. Consider now the product metric $\gamma=dx^2 +\gamma_{kerr}$ on $\mathbb{R}\times S^2$ in which $\gamma_{kerr}$ is the horizon metric from the near-horizon geometry of the extreme Kerr black hole. Let $h$ denote the natural extension to $\mathbb{R}\times S^2$ of the 1-form near-horizon data associated with extreme Kerr. By defining $F$ according to \eqref{1}, a solution $(F,h,\gamma)$ of the near-horizon geometry equations is produced on the universal cover. Since this set of near-horizon data is invariant under the $\mathbb{Z}_2\ast\mathbb{Z}_2$ action described above, by passing to the quotient a solution is obtained on $\mathbb{RP}^3 \# \mathbb{RP}^3$. It should be pointed out that this construction could be achieved by performing a similar quotient of an extreme Kerr string black hole, and taking the near-horizon limit. Thus, while it is not known whether this solution arises from an asymptotically flat parent black hole, it does arise from a quotient of an asymptotically Kaluza-Klein black hole.

What of \emph{nondegenerate} horizons? Should Conjecture \ref{conjecture9} also extend to them? In fact, it may be more interesting if the question were answered in the negative, so that degenerate horizons had a topological rigidity not seen in nondegenerate ones. One can then contemplate a quasi-stationary system having a horizon consisting of a nontrivial connected sum, being `spun up' adiabatically so that the mass is fixed and the system is always stationary to a good approximation. On approach to extremality, i.e. as the horizon nears a degenerate state, the horizon would have to undergo a dynamical instability, perhaps forming neckpinches to break apart the connected sum. The adiabatic approximation would likely fail at some point but, much worse, so could cosmic censorship. This picture is consistent with the 3rd law of black hole thermodynamics, which asserts that it is not possible to produce a black hole with vanishing surface gravity (temperature) through a physical process.
Any neckpinch region might well resemble a black string undergoing the Gregory-Laflamme instability \cite{GregoryLaflamme}. Of course, one possible resolution is that this scenario will not occur because such instabilities prevent (stationary vacuum) connected sum horizons from forming in the first place. This view might be seen as supporting the extension of Conjecture \ref{conjecture9}, except perhaps for unstable horizons. It's a question worth pursuing.

\medskip

\textbf{Acknowledgements.}
The authors would like to thank Hari Kunduri for comments, as well as Stefan Hollands and James Lucietti for independently pointing out that $\mathbb{RP}^3 \#\mathbb{RP}^3$ can be ruled out in Corollary \ref{cor7} when an extra $U(1)$ symmetry is present. The first author would also like to thank Luca Di Cerbo for useful discussions.

\end{document}